\documentclass[a4paper]{article}

\usepackage{a4wide}
\usepackage{amsmath,amssymb,amsthm}
\usepackage{mathrsfs,bbm}
\usepackage[american]{babel}
\usepackage[dvips]{graphicx}
\usepackage{array,pifont,varioref,framed}

\graphicspath{{./}{figure/}}

\title{A fully-discrete-state kinetic theory approach to modeling vehicular traffic}
\author{Luisa Fermo \\
		{\small\it Department of Mathematics and Computer Science} \\
		{\small\it University of Cagliari} \\
		{\small\it Viale Merello 92, 09123 Cagliari, Italy} \\[5mm]
		Andrea Tosin \\
		{\small\it Istituto per le Applicazioni del Calcolo ``M. Picone''} \\
		{\small\it Consiglio Nazionale delle Ricerche} \\		
		{\small\it Via dei Taurini 19, 00185 Rome, Italy}
	   }
\date{}

\newcommand{\1}{\mathbbm{1}}
\newcommand{\abs}[1]{\left\vert #1\right\vert}
\newcommand{\B}{\mathcal{B}}
\newcommand{\C}{\mathcal{C}}
\newcommand{\Dt}{\Delta{t}}

\newcommand{\f}{\mathbf{f}}
\newcommand{\g}{\mathbf{g}}
\newcommand{\fictrho}{\tilde{\rho}}
\newcommand{\In}{\textup{in}}
\newcommand{\kk}{\mathsf{k}}

\newcommand{\Lip}{\operatorname{Lip}}
\newcommand{\M}{\textup{max}}
\newcommand{\m}{\mathsf{m}}
\newcommand{\n}{\mathsf{n}}
\newcommand{\norm}[1]{\Vert #1\Vert}
\newcommand{\Out}{\textup{out}}
\newcommand{\R}{\mathbb{R}}
\renewcommand{\u}{\mathbf{u}}
\renewcommand{\v}{\mathbf{v}}

\theoremstyle{plain}\newtheorem{theorem}{Theorem}[section]
\theoremstyle{remark}\newtheorem{remark}[theorem]{Remark}
\theoremstyle{plain}\newtheorem{lemma}[theorem]{Lemma}
\theoremstyle{plain}\newtheorem{corollary}[theorem]{Corollary}
\theoremstyle{definition}\newtheorem{definition}[theorem]{Definition}
	
\begin{document}
\maketitle

\begin{abstract}
This paper presents a new mathematical model of vehicular traffic, based on the methods of the generalized kinetic theory, in which the space of microscopic states (position and velocity) of the vehicles is genuinely discrete. While in the recent literature discrete-velocity kinetic models of car traffic have already been successfully proposed, this is, to our knowledge, the first attempt to account for all aspects of the physical granularity of car flow within the formalism of the aforesaid mathematical theory. Thanks to a rich but handy structure, the resulting model allows one to easily implement and simulate various realistic scenarios giving rise to characteristic traffic phenomena of practical interest (e.g., queue formation due to roadworks or to a traffic light). Moreover, it is analytically tractable under quite general assumptions, whereby fundamental properties of the solutions can be rigorously proved.

\medskip


\medskip

\end{abstract}

\section{Introduction}
\label{sect.intro}
Modeling vehicular traffic at the kinetic scale offers several advantages in terms of capturing some of the relevant aspects of the complexity of car flow. The kinetic approach is indeed suitable for an aggregate representation of the distribution of vehicles, not necessarily focused on single cars, while still allowing for a detailed characterization of the microscopic vehicle-to-vehicle dynamics, which are ultimately responsible for the large-scale behavior of the system. In this respect, kinetic models are competitive over microscopic ones, see e.g., \cite{gazis1961nfl,helbing2001trs,treiber2006dia}. In fact, they require a lower number of equations, which makes them more accessible by computational and mathematical analysis as far as the study of average collective trends is concerned. On the other hand, unlike the prevailing approach at the macroscopic scale, see e.g., \cite{colombo2002hpt,lighthill1955kw2,richards1956swh}, kinetic models are not based on \emph{a priori} closures of the equations through fundamental diagrams. In particular, if necessary they can recover the latter \emph{a posteriori} out of a lower-scale modeling of microscopic interactions. Actually, it is worth pointing out that not all macroscopic models rely strictly on fundamental diagrams. Some models get rid of them by envisaging closures of the balance equations oriented to a phenomenological interpretation of the driver behavior \cite{aw2000rso,berthelin2008mfe}. For a more comprehensive overview of vehicular traffic models at all scales, the interested reader is referred to \cite{bellomo2011mtc}.

Classically, the kinetic representation of vehicular traffic along a one-way road is provided by a distribution function $f$ over the mechanical microscopic state of the vehicles. The latter is identified by the scalar position $x \in D_x$ and speed $v \in D_v$, where $D_x, D_v \subseteq \R$ are the spatial and speed domains, respectively. While the former may either be a bounded interval, such as $[0,\,L]$, $L>0$ being the length of the road, or coincide with the whole real axis, the latter is normally of the form $D_v=[0, V_{\M}]$, where $V_{\M}>0$ is the maximum speed allowed along the road, or possibly $D_v=[0, V_{\M}']$, with $V_{\M}'\geq V_{\M}$ corresponding to the maximum average speed  attainable by a single vehicle in free flow conditions. The distribution function $f=f(t,\,x,\,v):[0,T_{\M}] \times D_x\times D_v\to\R^+$, where $T_{\M}>0$ is the final time (possibly $T_{\M}=+\infty$), is then such that $f(t,\,x,\,v)dxdv$ is the (infinitesimal) number of vehicles that at time $t$ are located between $x$ and $x+dx$ and travel with a speed between $v$ and $v+dv$.

In the above presentation of the kinetic approach, the spatial position and speed of the vehicles are tacitly assumed to be continuously distributed over $D_x \times D_v$. However, this does not reflect correctly the physical reality of vehicular flow. Indeed, the number of vehicles along a road is normally not large enough for the continuity of the distribution function over the microscopic states to be an acceptable approximation (like in the classical kinetic theory of gases). Vehicles do not span continuously the whole set of admissible speeds; likewise, one cannot expect that they are so spread in space that at every point $x$ of the road there may be some. In other words, the actual distribution of vehicles in space, as well as that of their speeds, is strongly \emph{granular}. It is reasonable to expect this fact to have a nontrivial impact on the resulting dynamics. It is hence worth being taken into account in a mathematical model.

Recently discrete velocity models have been introduced \cite{bellouquid2012mvt,coscia2007mtv,delitala2007mmv}. The idea is to relax the hypothesis that the speed distribution is continuous  by introducing in the domain $D_v$ a lattice of discrete speeds $\{v_1,\,v_2,\,\dots,\,v_n\}$, with $v_1=0$, $v_n=V_{\M}$, $v_i<v_{i+1}$ for all $i=1,\,\dots,\,n-1$. The physical system is then described by $n$ distribution functions $f_j=f_j(t,\,x):[0,\,T_{\M}]\times D_x\to\R^+$, $j=1,\,\dots,\,n$, such that $f_j(t,\,x)=f(t,\,x,\,v_j)$ or, in distributional sense, $f(t,\,x,\,v)=\sum_{j=1}^{n} f_j(t,\,x)\delta_{v_j}(v)$. Particularly, $f_j(t,\,x) dx$ is the (infinitesimal) number of vehicles traveling at speed $v_j$, that, at time $t$, occupy a position comprised between $x$ and $x+dx$. In this way, the granular character of the car flow is at least partially taken into account from the point of view of the speed distribution.

Starting from the discrete velocity kinetic framework developed in the above-cited papers, in the present work we undertake the discretization also of the microscopic space variable, in order to accomplish the program of a fully-discrete-state kinetic theory of vehicular traffic. The basic idea is to partition the spatial domain $D_x$ in disjoint cells (subintervals) $I_i$ such that $D_x =\cup_i I_i$. Each cell identifies a portion of the road in which only the number of vehicles is known. No matter how rough this assumption may seem, it is nevertheless realistic in view of the spatial granularity of traffic, which does not allow one to make finer predictions on the spatial position of vehicles within each cell. The microscopic state $(x,v)$ of the vehicles belongs then to the discrete state space $\{I_i\}_i\times\{v_j\}_j$, which generates a new structure of the kinetic equations not immediately deducible from more standard frameworks. Particularly, it is worth anticipating that the resulting mathematical structure is not a cellular automaton \cite{daganzo2006tfc,deutsch2005cam} in spite of the discreteness of the space state. In fact, as it will be clear in the next sections, vehicles are not assimilated to fictive particles jumping from their current site to a neighboring one with prescribed probability. They actually flow through the cells with their true speed, according to a transport term duly implemented in the time-evolution equations of their distribution functions. Car-to-car interactions, responsible for speed variations, are in turn coded in the aforesaid equations, so that finally the evolution of the system is not seen as a stepwise algorithmic update of the lattice of microscopic states.

In developing the program just outlined, this paper aims at covering the whole path from modeling to numerical simulations to well-posedness analysis of the resulting mathematical problems. In more detail, Section~\ref{sect:kin.fram} derives the general discrete-state equations, resting on the methods of the generalized kinetic theory \cite[Sect.~6]{bellomo2011mtc}, to be used as a reference framework for specific models. This framework accounts, in particular, for stochastic dynamics of speed variations, according to the idea that one-to-one microscopic interactions among vehicles are assimilated to stochastic games. Section~\ref{sect:model} presents a model obtained from the previous framework, which is then used in Section~\ref{sect:comp.anal} for producing exploratory numerical simulations of some case studies. These include the long-term trend of the system, particularly with reference to the fundamental diagrams predicted by the kinetic equations, and the dynamics of queue formation/depletion along a road triggered by two different causes: roadworks and traffic light. Finally, Section~\ref{sect:qual.anal} is devoted to the analytical study of some relevant qualitative properties of the initial/boundary-value problems. Under quite general assumptions, to be possibly regarded as modeling guidelines, global existence of the solutions, uniqueness, continuous dependence on the initial and boundary data, non-negativity, boundedness are proved, which provides the modeling framework with the necessary amount of mathematical robustness.

\section{Discrete space kinetic framework}
\label{sect:kin.fram}
In this section, starting from the discrete velocity setting recalled in the Introduction, we derive kinetic structures based on a fully discrete space of the mechanical microscopic states of the vehicles.

As previously anticipated, partitioning the spatial domain $D_x$ in a number of cells of finite size is a more realistic way of detecting the positions of the vehicles along the road. Indeed, it is consistent with the intrinsic granularity of the flow, which does not allow for a statistical description of their spatial distribution more accurate than a certain minimum level of detail. In addition, as we will see, it enables one to account easily for some effects due to the finite size of the vehicles even in a context in which the actual representation is not focused on each of them. To this purpose, a useful partition of the spatial domain $D_x$ is in pairwise disjoint cells $I_i$, whose union is $D_x$. For a road $D_x=[0,\,L]$ of length $L>0$, we set
$$ D_x=\bigcup_{i=1}^{m}I_i, \qquad I_{i_1}\cap I_{i_2}=\emptyset, \quad \forall\,i_1\ne i_2, $$
$m\in\mathbb{N}$ being the total number of cells needed for covering $D_x$, which depends on the size $\ell_i$ of each $I_i$. We will henceforth assume that cells have a constant size $\ell$, chosen in a such a way that $L/\ell\in\mathbb{N}$, thus $m=L/\ell$.

Let $f_{ij}=f_{ij}(t)$ be the distribution function of vehicles that, at time $t$, are located in the $i$-th cell with a speed in the $j$-th class. Then the total number $N_{ij}$ of vehicles in $I_i$ with speed $v_j$ is $N_{ij}=f_{ij}\ell$, while the total number $N_i$ of vehicles in $I_i$ is
$$ N_i=\sum_{j=1}^{n}N_{ij}=\ell\sum_{j=1}^{n}f_{ij}. $$
Notice that $\rho_i:=N_i/\ell$ is instead the local density of vehicles in the $i$-th cell. By further summing over $i$, one gets the total number $N$ of vehicles along the road:
\begin{equation}
	N=\sum_{i=1}^{m}N_i=\ell\sum_{i=1}^{m}\sum_{j=1}^{n}f_{ij}.
	\label{eq:N}
\end{equation}

\begin{remark}
If $N_{\M}>0$ is the maximum number of vehicles allowed in a single cell, then the maximum local density is $\rho_{\M}:=N_{\M}/\ell$, which defines the \emph{road capacity}.
\end{remark}

From the $f_{ij}$'s, the kinetic distribution function $f$ can be recovered as
\begin{equation}
	f(t,\,x,\,v)=\sum_{i=1}^{m}\sum_{j=1}^{n}f_{ij}(t)\1_{I_i}(x)\delta_{v_j}(v),
	\label{eq:f}
\end{equation}
$\1_{I_i}$ being the characteristic function of the cell $I_i$ ($\1_{I_i}(x)=1$ if $x\in I_i$, $\1_{I_i}(x)=0$ if $x\not\in I_i$). In practice, $f$ is an atomic distribution with respect to the variable $v$, like in the discrete velocity framework, and is piecewise constant with respect to the variable $x$. Particularly, this latter characteristic implies that vehicles are thought of as uniformly distributed within each cell.

Usual macroscopic variables of traffic, such as the vehicle density $\rho$, flux $q$, and average speed $u$ are obtained from Eq.~\eqref{eq:f} as distributional moments of $f$ with respect to $v$:
\begin{equation}
	\rho(t,\,x)=\sum_{i=1}^{m}\left(\sum_{j=1}^{n}f_{ij}(t)\right)\1_{I_i}(x), \quad
		q(t,\,x)=\sum_{i=1}^{m}\left(\sum_{j=1}^{n}v_j f_{ij}(t)\right)\1_{I_i}(x), \quad
			u(t,\,x)=\frac{q(t,\,x)}{\rho(t,\,x)}.
	\label{eq:macro.var}
\end{equation}
Notice that $\rho$ can also be written as $\rho(t,\,x)=\sum_{i=1}^m\rho_i(t)\1_{I_i}(x)$, thus showing that the vehicle density on the whole road is a piecewise constant function in $x$ built from the local cell densities. Analogously, after defining the flux at the $i$-th cell as $q_i:=\sum_{j=1}^{n}v_jf_{ij}$, it results $q(t,x)=\sum_{i=1}^{m}q_i(t)\1_{I_i}(x)$.

\subsection{Evolution equations for the $\boldsymbol{f_{ij}}$'s}
A mathematical model  of vehicular traffic in a fully discrete kinetic context is obtained by deriving suitable evolution equations for the distribution functions $f_{ij}$. To this end, we propose a method grounded on classical ideas of conservation laws, which however  exploits directly the discrete space structure as a distinctive feature of the present setting.

Fix a cell $I_i$ and a velocity class $j$. The basic idea is, as usual, that during a time span $\Dt>0$ the number $N_{ij}$ of vehicles in that cell with that velocity may vary because of:
\begin{itemize}
\item vehicles leaving the cell $I_i$ or new vehicles entering it at speed $v_j$ from neighboring cells;
\item interactions among vehicles within the cell $I_i$ inducing speed changes.
\end{itemize}
On the basis of these two principles, the following preliminary balance can be written:
\begin{equation}
	N_{ij}(t+\Dt)-N_{ij}(t)=-N_{ij}^\Out([t,\,t+\Dt])+N_{ij}^\In([t,\,t+\Dt])+J_{ij}(t)\ell\Dt+ o(\Dt)
	\label{eq:balance.N}
\end{equation}
with obvious meaning of the symbols $N_{ij}^\In$, $N_{ij}^\Out$. The term $J_{ij}(t)$ at the right-hand side accounts for instantaneous acceleration/braking dynamics per unit space $\ell$, that produce transitions among the various speed classes. For the moment we keep it generic, deferring to the next section a detailed discussion on its structure. Moreover, the term $o(\Delta t)$ indicates that we are neglecting higher order effects which may contribute to the variation of $N_{ij}$ during the time $\Dt$.

If the time span $\Dt$ is sufficiently small, we can assume that the number of vehicles leaving the $i$-th cell depends solely on the number of vehicles already present in $I_i$ at time $t$. To be definite, let $v_j\Dt\leq\ell$ for all $j$, so that no vehicle can move forward  for more than one cell in a single $\Dt$. Then the following form of $N_{ij}^\Out$ can be proposed:
$$ N_{ij}^\Out([t,\,t+\Dt])=\frac{v_j\Dt}{\ell}\Phi_{i,i+1}N_{ij}(t)+o(\Delta{t}), $$
where:
\begin{itemize}
\item the ratio $v_j\Dt/\ell$ accounts for the percent amount of vehicles traveling at speed $v_j$, which are sufficiently close to the boundary between the $i$-th cell (whence they are coming) and the $(i+1)$-th cell (where they are going) for changing cell during the time span $\Dt$;
\item the term $\Phi_{i,i+1}$, that we name \emph{flux limiter} across the interface separating the cells $I_i$, $I_{i+1}$, accounts for the percent amount of vehicles in $I_i$ that can actually move to $I_{i+1}$ on the basis of the free space available in the latter.
\end{itemize}
Likewise, the following form of $N_{ij}^\In$ can be proposed:
$$ N_{ij}^\In([t,\,t+\Dt])=\frac{v_j\Dt}{\ell}\Phi_{i-1,i}N_{i-1,j}(t)+o(\Delta{t}), $$
where now the space dynamics involve the $(i-1)$-th and the $i$-th cells.

Plugging these expressions in Eq.~\eqref{eq:balance.N} and rewriting in terms of the $f_{ij}$'s yields
\begin{equation}
	f_{ij}(t+\Dt)\ell-f_{ij}(t)\ell=-\frac{v_j\Dt}{\ell}\left(\Phi_{i,i+1}f_{ij}(t)\ell-\Phi_{i-1,i}f_{i-1,j}(t)\ell\right)
		+J_{ij}(t)\ell\Dt+o(\Dt)
	\label{dis.eq.balance}
\end{equation}
whence, dividing by $\ell\Dt$, taking the limit $\Dt\to 0^+$, and rearranging the terms produces
\begin{equation}
	\frac{df_{ij}}{dt}+\frac{v_j}{\ell} (\Phi_{i,i+1}f_{ij}-\Phi_{i-1,i}f_{i-1,j})=J_{ij}.
	\label{eq:balance.fij}
\end{equation}

This equation is in principle well defined only in the internal cells of the domain $D_x$, namely for $2 \leq i \leq m-1$. Indeed, at the boundary cells $i=1$, $i=m$ the flux limiters $\Phi_{0,1}$, $\Phi_{m,m+1}$ call for the two nonexistent external cells $I_0$, $I_{m+1}$. In addition, for $i=1$ the equation requires also the value $f_{0j}$. In order to overcome these difficulties, we consider that the condition $v_j \geq 0$ for all $j$ implies a unidirectional (rightward) flux of vehicles, which requires thus a condition on the left boundary of the domain. Consequently, the values $\Phi_{0,1}$ and $f_{0j}$ have to be provided as boundary conditions. On the contrary, the flux limiter $\Phi_{m,m+1}$ has to be prescribed at the right boundary in order to specify how the natural outflow of the vehicles is possibly modified by external conditions. In standard situations, it will be simply $\Phi_{m,m+1}=1$.

\begin{remark}
The formal derivation of Eq.~\eqref{eq:balance.fij} closely reminds of the standard one of conservation laws: the time variation of the conserved quantity is related to some drift plus microscopic conservative interactions typical of kinetic theory. Nevertheless, it is worth pointing out that classical procedures, such as the straightforward use of the divergence theorem, for deducing conservation laws in differential form cannot be applied in the present context. Indeed, the finite size of the space cells is a structural feature of the problem at hand, which prevents one from taking the limit $\ell\to 0^+$.
\end{remark}

\begin{remark}[Conservativeness of Eq.~\eqref{eq:balance.fij}]
The total number $N$ of vehicles along the road, defined by Eq.~\eqref{eq:N}, has to balance with the inflow/outflow of vehicles from the boundaries of $D_x$. We can formally check this property out of Eq.~\eqref{eq:balance.fij}, after recalling that the conservative interaction operator $J_{ij}$ is required to fulfill
\begin{equation}
	\sum_{j=1}^{n}J_{ij}(t)=0, \quad \forall\,i=1,\dots,m \  \,, t>0.
	\label{eq:cons.J}
\end{equation}
Multiplying Eq.~\eqref{eq:balance.fij} by $\ell$ and summing both sides over $i$, $j$ gives then
$$ \frac{dN}{dt}+\sum_{j=1}^{n}v_j\left[\sum_{i=1}^{m}\left(\Phi_{i,i+1}f_{ij}
	-\Phi_{i-1,i}f_{i-1,j}\right)\right]=0, $$
whence, summing telescopically in the square brackets, yields the desired result:
\begin{align*}
\frac{dN}{dt} &= -\Phi_{m,m+1}\sum_{j=1}^{n}v_jf_{mj}+\Phi_{0,1}\sum_{j=1}^{n}v_jf_{0j} \\
			  &= -\Phi_{m,m+1}q_m+\Phi_{0,1}q_0,
\end{align*}
$q_0$, $q_m$ being the incoming and outgoing fluxes of vehicles through the boundaries $x=0$ and $x=L$, respectively.
\end{remark}

\subsection{Stochastic dynamics of speed transitions}
The right-hand side of Eq.~\eqref{eq:balance.fij} is the phenomenological core of the kinetic equations, because it accounts for the time variation of the conserved quantities $f_{ij}$ in terms of microscopic interactions among vehicles. The operator $J_{ij}$ is called \emph{interaction operator}. It describes modifications of the speed $v_j$ of a vehicle traveling in the cell $I_i$ due to the action exerted on it by other vehicles.

Vehicles are not purely mechanical particles, for the presence of drivers provides them with the ability of taking decisions \emph{actively}, i.e., without necessarily the influence of external force fields. Consequently, the description delivered by $J_{ij}$ cannot be purely deterministic: personal behaviors have to be taken into account, which are suitably modeled from a stochastic point of view. Specifically, interactions among vehicles are assimilated to \emph{stochastic games} among couples of players. The game strategy of each player, viz. vehicle, is represented by its speed, while the payoff of the game is the new speed class it shifts to after the interaction. Games are stochastic because payoffs are known only in probability. In other words, the model considers that vehicles might not react always the same when placed in the same conditions due to possible partly subjective decisions of the drivers.

These ideas are brought to a formal level by writing $J_{ij}$ as a balance, in the space of microscopic states $\{I_i\}_{i=1}^{m}\times\{v_j\}_{j=1}^{n}$, of vehicles gaining and losing the state $(I_i,\,v_j)$ in the unit time:
\begin{equation}
	J_{ij}=\frac{\ell}{2}\left(\sum_{h,\,k=1}^{n}\eta_{hk}(i)A_{hk}^{j}(i)f_{ih}f_{ik}-
		f_{ij}\sum_{k=1}^{n}\eta_{jk}(i)f_{ik}\right).
	\label{eq:J}
\end{equation}
This expression is formally obtained by integrating, over a single space cell $I_i$, the interaction operator deduced in previous works dealing with the discrete-velocity kinetic theory of vehicular traffic (see e.g., \cite{tosin2009fgk}). In this sense, Eq.~\eqref{eq:J} is consistent with other models relying on similar theoretical backgrounds.

We mention that:
\begin{itemize}
\item Only binary interactions among vehicles are accounted for, which possibly produce a change of speed but not of position of the interacting pairs. In particular, it is customary to term \emph{candidate} the vehicle which is likely to modify its current speed $v_h$ into the \emph{test} speed $v_j$ (the payoff of the game) after an interaction with a \emph{field} vehicle with speed $v_k$.
\item The quantitative description of the interactions is provided by the terms $\boldsymbol{\eta}=\{\eta_{hk}(i)\}_{h,\,k=1,\,\dots,\,n}^{i=1,\,\dots,\,m}$, $\boldsymbol{A}=\{A_{hk}^j(i)\}_{h,\,j,\,k=1,\,\dots,\,n}^{i=1,\,\dots,\,m}$ called the \emph{interaction rate} and the \emph{table of games}, respectively. The former models the frequency of the interactions among candidate and field vehicles. The latter gives instead the probabilities that candidate vehicles get the test speeds after interacting with field vehicles (hence $A_{hk}^j$ is, for $j=1,\,\dots,\,n$, the probability distribution of the candidate player's payoff conditioned to the strategies $v_h$, $v_k$ by which the two involved players approach the game). Notice that both $\eta_{hk}(i)$ and $A_{hk}^{j}(i)$ may vary from cell to cell. Since the table of games is a discrete probability distribution, the following conditions must hold:
\begin{align}
	0\leq A_{hk}^{j}(i)\leq 1, \qquad & \forall\,h,\,j,\,k=1,\,\dots,\,n, \quad \forall\,i=1,\,\dots,\,m \nonumber \\
	\sum_{j=1}^{n}A_{hk}^{j}(i)=1, \qquad & \forall\,h,\,k=1,\,\dots,\,n, \quad \forall\,i=1,\,\dots,\,m. \label{eq:tog.sum.1}
\end{align}
Particularly, from Eq.~\eqref{eq:tog.sum.1} it follows the fulfillment of Eq.~\eqref{eq:cons.J}.
\item The coefficient $\ell/2$ appearing in expression \eqref{eq:J} weights the above interactions over the characteristic length of a single space cell.
\end{itemize}

\begin{remark}[Generalization to nonlocal interactions]
The form \eqref{eq:J} of the interaction operator tacitly implies that candidate and field vehicles stay within the same space cell. However, one can also assume that the candidate vehicle interacts with field vehicles located in further cells ahead. This amounts to introducing a set of \emph{interaction cells}
$$ \{I_i,\,I_{i+1},\,\dots,\,I_{i+\mu(i)}\} $$
extending from the $i$-th cell up to $\mu(i)$ cells in front. The number $\mu(i)$ varies in general with the index $i$ so as to satisfy the constraint $i+\mu(i)\leq m$ (i.e., interaction cells cannot lie beyond the last cell $I_m$). Consequently, the structure \eqref{eq:J} of $J_{ij}$ modifies as
\begin{equation}
	J_{ij}=\frac{\ell}{2}\left(\sum_{l=i}^{i+\mu(i)}\sum_{h,\,k=1}^{n}\eta_{hk}(l)
		A_{hk}^{j}(l)f_{ih}f_{lk}w_l-f_{ij}\sum_{l=i}^{i+\mu(i)}\sum_{k=1}^{n}
			\eta_{jk}(l)f_{lk}w_l\right),
	\label{eq:J.nonlocal}
\end{equation}
where $\{w_l\}_{l=i}^{i+\mu(i)}$ are suitable weights accounting for different effects of the interactions based on the distance of the $l$-th cell from the $i$-th one. They have to comply with some minimal requirements:
\begin{equation}
	w_l\geq 0 \quad \forall\,l=i,\,\dots,\,i+\mu(i), \qquad	\sum_{l=i}^{i+\mu(i)}w_l=1 \quad \forall\,i=1,\,\dots,\,m.
	\label{eq:w.sum.1}
\end{equation}
Notice that, for $\mu\equiv 0$, Eq.~\eqref{eq:w.sum.1} implies that the only remaining weight $w_i$ is equal to $1$, which ultimately enables one to recover Eq.~\eqref{eq:J} as a special case of Eq.~\eqref{eq:J.nonlocal}.
\end{remark}

\subsection{Non-dimensionalization}
At this point, it is useful to rewrite Eq.~\eqref{eq:balance.fij} in dimensionless form. To this end, we choose $\ell$ and $V_\M$ as characteristic values of length and speed, respectively, and we define the following nondimensional variables and functions:
$$ v_j^\ast:=\frac{v_j}{V_\M}, \qquad t^\ast:=\frac{V_\M}{\ell}t, \qquad
	f_{ij}^\ast(t^\ast):=\frac{\ell}{N_\M}f_{ij}\left(\frac{\ell}{V_\M}t^\ast\right). $$
Substituting these expressions into Eqs.~\eqref{eq:balance.fij}, \eqref{eq:J}, we get, after some algebraic manipulations,
\begin{equation}
	\frac{df_{ij}^\ast}{dt^\ast}+v_j^\ast\left(\Phi_{i,i+1}f_{ij}^\ast
		-\Phi_{i-1,i}f^\ast_{i-1,j}\right)=\sum_{h,\,k=1}^{n}\eta_{hk}^\ast(i)A_{hk}^{j}(i)
			f_{ih}^\ast f_{ik}^\ast-f_{ij}^\ast\sum_{k=1}^{n}\eta_{jk}^\ast(i)f_{ik}^\ast,
	\label{eq:balance.fij.complete}
\end{equation}
where
$$ \eta_{hk}^\ast(i):=\frac{\ell N_\M}{2V_\M}\eta_{hk}(i) $$
is the dimensionless interaction rate.

It is worth noting that the new speed lattice is now
$$ v_1^\ast=0 < \dots < v_i^\ast < v_{i+1}^\ast < \dots < v_{n}^\ast=1, $$
while the nondimensional length of the space cells becomes unitary. In addition, one can check that the dimensionless local cell density reads
\begin{equation}
	\rho_i^\ast(t^\ast):=\frac{\ell}{N_\M}\rho_i\left(\frac{\ell}{V_\M}t^\ast\right)
		=\frac{N_i\left(\frac{\ell}{V_\M}t^\ast\right)}{N_\M},
	\label{eq:rho.nondim}
\end{equation}
whence $0\leq\rho_i^\ast\leq 1$. Consequently, $\rho_i^\ast$ can be directly thought of as a probability density.

We will henceforth refer to Eq.~\eqref{eq:balance.fij.complete}, omitting the asterisks on the dimensionless quantities for brevity.

\section{From the general framework to particular models}
\label{sect:model}
Starting from the mathematical structures given in Sect.~\ref{sect:kin.fram}, specific models can be produced by detailing the flux limiters, the interaction rate, and the table of games. This amounts to analyzing the dynamics of microscopic interactions occurring among vehicles.

\paragraph*{Flux limiters} We recall that the term $\Phi_{i,i+1}$ limits the number of vehicles that can actually travel across the cells on the basis of the occupancy of the destination cells.

A prototypical form of such a limiter can be derived via the following argument. The free space left in the cell $I_{i+1}$ for vehicles coming from the cell $I_i$ is $N_{\M}-N_{i+1}$. Assuming $N_i>N_{\M}-N_{i+1}$, i.e., that the number of vehicles in the $i$-th cell is greater than such an available space, the percent amount of vehicles that can actually move is
$$ \Phi_{i,i+1}=\frac{N_{\M}-N_{i+1}}{N_i}. $$
On the other hand, if $N_i \leq N_{\M}-N_{i+1}$ then in the cell $I_{i+1}$ there is, in principle, enough room for all vehicles coming from the cell $I_i$, thus in this case $\Phi_{i,i+1}=1$.

Rewriting the flux limiter in terms of the local cell density via Eq.~\eqref{eq:rho.nondim}, it results
\begin{equation}
	\Phi_{i,i+1}=
	  \begin{cases}
    	\dfrac{1-\rho_{i+1}}{\rho_i} & \text{if\ } \rho_i+\rho_{i+1}>1 \\
	    1 & \text{if\ } \rho_i+\rho_{i+1}\leq 1.
	  \end{cases}
  \label{eq:flux.lim}
\end{equation}
The interpretation is that if the total density in the pair of cells $(I_i,\,I_{i+1})$ is greater than the (dimensionless) road capacity then the microscopic granularity of traffic starts acting and all vehicles cannot freely flow through the cells. Otherwise, the flux limiter has actually no effect.

\paragraph*{Interaction rate} The frequency of the interactions can be assumed to depend on the number of vehicles in each cell, in such a way that the larger this number the higher the frequency. In view of the indistinguishability of the vehicles, a possible simple form implementing this idea is
\begin{equation}
	\eta_{hk}(i)=\eta_0\frac{N_i}{N_\M}=\eta_0\rho_i,
	\label{eq:eta}
\end{equation}
where $\eta_0>0$ is a basic interaction frequency. Notice that this specific form does not depend explicitly on the speed classes of the interacting pairs. From now on we will invariably consider interaction rates of this kind and we will use accordingly the simplified notation $\eta(i)$.

\paragraph*{Table of games} Speed transitions within the space cells depend, in general, on the evolving traffic conditions, namely, among others, the local car congestion and the quality of the environment (e.g., road or weather). In this paper, we account for the former via the concept of \emph{fictitious density} explained below, which aims at introducing in the table of games the \emph{active} ability of drivers to anticipate the actual evolution of traffic. Conversely, we identify the latter simply by a parameter $\alpha\in [0,\,1]$, $\alpha=0$ standing for the worst environmental conditions and $\alpha=1$ for the best ones.

The concept of fictitious density, originally introduced in \cite{deangelis1999nhm}, is here revisited in the frame of the discrete space structure and of the stochastic game approach to vehicle interactions. We define the fictitious car density $\fictrho_i$ felt by drivers in the cell $I_i$ as the following weighted average of the actual car densities in the cells $I_i$, $I_{i+1}$:
\begin{equation}
	\fictrho_i:=
		\begin{cases}
			(1-\beta)\rho_i+\beta\rho_{i+1} & \text{for\ } i=1,\,\dots,\,m-1 \\
			\rho_m & \text{for\ } i=m,
		\end{cases}
	\label{eq:fictrho}
\end{equation}
where $\beta\in [0,\,1]$ is a parameter related to the anticipation ability of drivers. From Eq.~\eqref{eq:fictrho} it results that if $\rho_{i+1}\leq\rho_i$ then $\fictrho_i\leq\rho_i$, thus if the car density in the cell ahead is lower than the one in the current cell drivers may be motivated to behave as if the traffic congestion in their cell were lower than the actual one. This implies, for instance, a higher inclination to maintain the speed, or even to accelerate, in spite of a possibly insufficient local free space, because they are anticipating the nearby evolution of traffic. Conversely, if $\rho_{i+1}>\rho_i$ then it results $\fictrho_i>\rho_i$, hence if the car density in the cell ahead is higher than the one in the current cell drivers may be motivated to keep a more precautionary behavior, for instance by decelerating in spite of the real local free space. Notice that $\fictrho_i=\rho_i$ if and only if $\beta=0$ or $\rho_{i+1}=\rho_i$. In these cases, driver's anticipation ability is countered by either a specifically modeled behavior or the instantaneous traffic evolution.

The table of games stemming from these ideas is described in the following. Three cases need to be dealt with, corresponding to candidate vehicles $h$ traveling more slowly, or faster, or at the same speed as field vehicles $k$.

\begin{itemize}
\item $v_h<v_k$. In this case, the candidate vehicle may either decide to maintain its speed $v_h$ or be motivated by the faster field vehicle to accelerate to the speed $v_{h+1}$. The better the environmental conditions ($\alpha$), and the more the free room felt by drivers in the $i$-th cell ($1-\fictrho_i$), the higher the probability to accelerate. Such actions are further modulated by the actual free room available in the destination cell $I_{i+1}$. This effect is accounted for by the flux limiter $\Phi_{i,i+1}$, specifically by the fact that the candidate vehicle can be forced to stop (corresponding to the speed transition $v_h\to v_1=0$) as the next cell fills. A technical modification is necessary at the boundary $h=1$ of the speed lattice, for if the candidate vehicle is already still then the options of either maintaining the speed or being forced to stop by the flux limiter coincide.
\begin{equation}
	v_h<v_k
	\begin{cases}
		h=1 &
			\begin{cases}
				A_{1k}^1(i)=1-\alpha(1-\fictrho_i)\Phi_{i,i+1} \\
				A_{1k}^2(i)=\alpha(1-\fictrho_i)\Phi_{i,i+1} \\
				A_{1k}^j(i)=0 \quad \text{otherwise}
			\end{cases}	\\
		h>1 &
			\begin{cases}
				A_{hk}^1(i)=1-\Phi_{i,i+1} \\
				A_{hk}^h(i)=[1-\alpha(1-\fictrho_i)]\Phi_{i,i+1} \\
				A_{hk}^{h+1}(i)=\alpha(1-\fictrho_i)\Phi_{i,i+1} \\
				A_{hk}^j(i)=0 \quad \text{otherwise}
			\end{cases}
	\end{cases}
	\label{eq:tab.games.h<k}
\end{equation}
\item $v_h>v_k$. In this case, the candidate vehicle may either maintain its speed $v_h$, if e.g., it can overtake the leading field vehicle, or be forced to queue by decelerating to the speed $v_k$ of the latter. Specifically, from the expression of the corresponding transition probability reported in Eq.~\eqref{eq:tab.games.h>k} below, it can be noticed that this latter event is fostered by bad environmental conditions as well as by a reduced free room felt by drivers in the $i$-th cell. In particular, if the candidate vehicle interacts with a stationary field vehicle ($k=1$), the transition probabilities are deduced by merging, like before, the first two cases of the expression valid for $k>1$.
\begin{equation}
	v_h>v_k
	\begin{cases}
		k=1 &
			\begin{cases}
				A_{h1}^1(i)=1-\alpha(1-\fictrho_i)\Phi_{i,i+1} \\
				A_{h1}^h(i)=\alpha(1-\fictrho_i)\Phi_{i,i+1} \\
				A_{h1}^j(i)=0 \quad \text{otherwise}
			\end{cases} \\
		k>1 &
			\begin{cases}
				A_{hk}^1(i)=1-\Phi_{i,i+1} \\
				A_{hk}^k(i)=[1-\alpha(1-\fictrho_i)]\Phi_{i,i+1} \\
				A_{hk}^h(i)=\alpha(1-\fictrho_i)\Phi_{i,i+1} \\
				A_{hk}^j(i)=0 \quad \text{otherwise}
			\end{cases}
	\end{cases}
	\label{eq:tab.games.h>k}
\end{equation}
\item $v_h=v_k$. Finally, in this case the candidate vehicle can either be motivated to accelerate to the speed $v_{h+1}$ or induced to decelerate to the speed $v_{h-1}$ by the field vehicle; alternatively, it may also decide to maintain its speed $v_h$. The corresponding transition probabilities are built from arguments analogous to those discussed in the previous points, thus invoking again the concepts of environmental conditions, free room felt by drivers in the $i$-th cell, and modulation effect by the flux limiter. At the boundaries of the speed lattice some technical modifications are needed, for some output events may coincide while others may not apply (e.g., deceleration and acceleration are not possible in the lowest and highest speed classes, respectively).
\begin{equation}
	v_h=v_k
	\begin{cases}
		h=1 &
			\begin{cases}
				A_{11}^1(i)=1-\alpha(1-\fictrho_i)\Phi_{i,i+1} \\
				A_{11}^2(i)=\alpha(1-\fictrho_i)\Phi_{i,i+1} \\
				A_{11}^j(i)=0 \quad \text{otherwise}
			\end{cases} \\
		h=2 &
			\begin{cases}
				A_{22}^1(i)=1-\Phi_{i,i+1}+(1-\alpha)\fictrho_i\Phi_{i,i+1} \\
				A_{22}^2(i)=[1-\alpha-(1-2\alpha)\fictrho_i]\Phi_{i,i+1} \\
				A_{22}^3(i)=\alpha(1-\fictrho_i)\Phi_{i,i+1} \\
				A_{22}^j(i)=0 \quad \text{otherwise}
			\end{cases} \\
		2<h<n &
			\begin{cases}
				A_{hh}^1(i)=1-\Phi_{i,i+1} \\
				A_{hh}^{h-1}(i)=(1-\alpha)\fictrho_i\Phi_{i,i+1} \\
				A_{hh}^h(i)=[1-\alpha-(1-2\alpha)\fictrho_i]\Phi_{i,i+1} \\
				A_{hh}^{h+1}(i)=\alpha(1-\fictrho_i)\Phi_{i,i+1} \\
				A_{hh}^j(i)=0 \quad \text{otherwise}
			\end{cases} \\
		h=n &
			\begin{cases}
				A_{nn}^1(i)=1-\Phi_{i,i+1} \\
				A_{nn}^{n-1}(i)=(1-\alpha)\fictrho_i\Phi_{i,i+1} \\
				A_{nn}^n(i)=[1-(1-\alpha)\fictrho_i]\Phi_{i,i+1} \\
				A_{nn}^j(i)=0 \quad \text{otherwise}
			\end{cases}
	\end{cases}
	\label{eq:tab.games.h=k}
\end{equation}
\end{itemize}

\section{Computational analysis}
\label{sect:comp.anal}
The validity of the model presented in the previous sections, and indirectly also that of the methodological approach which generated it, can be assessed through exploratory numerical simulations addressing typical phenomena of vehicular traffic.

\subsection{The spatially homogeneous problem}
A first common case study refers to car flow in uniform space conditions, i.e., when one assumes that cars are homogeneously distributed along the road. As a matter of fact, this ideal scenario is the tacit assumption behind the empirical study of the so-called \emph{fundamental diagrams} out of experimental data measured for real flows. Such diagrams relate the density of cars to either their average speed or their flux, this way providing synthetic insights into the gross phenomenology of car flow expected in stationary conditions. The representation at the mesoscopic (kinetic) scale is particularly suited to simulate, by mathematical models, this kind of information. In fact, the distribution function allows one to obtain genuinely the required statistics (cf. Eq.~\eqref{eq:macro.var}) from the modeled microscopic interactions, rather than relying on them for modeling directly the average macroscopic dynamics. Specifically, considering the basic assumption of spatial homogeneity, the equations to be solved are:
\begin{equation}
	\frac{df_j}{dt}=\eta\left(\sum_{h,\,k=1}^{n}A_{hk}^jf_hf_k-f_j\rho\right), \qquad j=1,\,\dots,\,n,
	\label{eq:spat.homog}
\end{equation}
which are derived from Eq.~\eqref{eq:balance.fij.complete} with $f_{ij}=f_{i+1,j}$ for $i=1,\,\dots,\,m-1$ and all $j$, along with the assumption that the interaction rate does not depend on $h$, $k$, cf. Eq.~\eqref{eq:eta}. In view of Eq.~\eqref{eq:flux.lim}, also the flux limiters turn out to be independent of $i$, which makes the transport term at the left-hand side of Eq.~\eqref{eq:balance.fij.complete} vanish. Ultimately, the index $i$ of the spatial cell has been removed because it is unnecessary in the present context.

Owing to Eq.~\eqref{eq:tog.sum.1}, Eq.~\eqref{eq:spat.homog} is such that the car density $\rho=\sum_{j=1}^{n}f_j$ is conserved in time:
\begin{equation*}
	\frac{d\rho}{dt}=\eta\left(\sum_{h,\,k=1}^{n}f_hf_k-\rho\sum_{j=1}^{n}f_j\right)=0,
\end{equation*}
therefore one can fix $\rho\in[0,\,1]$ at the initial time and then, by means of Eq.~\eqref{eq:spat.homog}, look for the asymptotic statistical distribution (if any) of microscopic speeds corresponding to the very same car density. Following this procedure, the fundamental diagrams of Fig.~\ref{fig:fund.diag} have been numerically obtained for various values of the parameter $\alpha$ representing the environmental conditions in the table of games (cf. Eqs.~\eqref{eq:tab.games.h<k}--\eqref{eq:tab.games.h=k}). It is worth noticing that the other parameters $\eta_0$ (cf. Eq.~\eqref{eq:eta}) and $\beta$ (cf. Eq.~\eqref{eq:fictrho}) play instead no role in shaping the asymptotic distributions, thus also the fundamental diagrams. Indeed, the first one can be hidden in the time scale whereas the second one drops because in spatially homogeneous conditions the fictitious car density coincides with the actual one.

\begin{figure}[!t]
\centering
\includegraphics[width=0.95\textwidth]{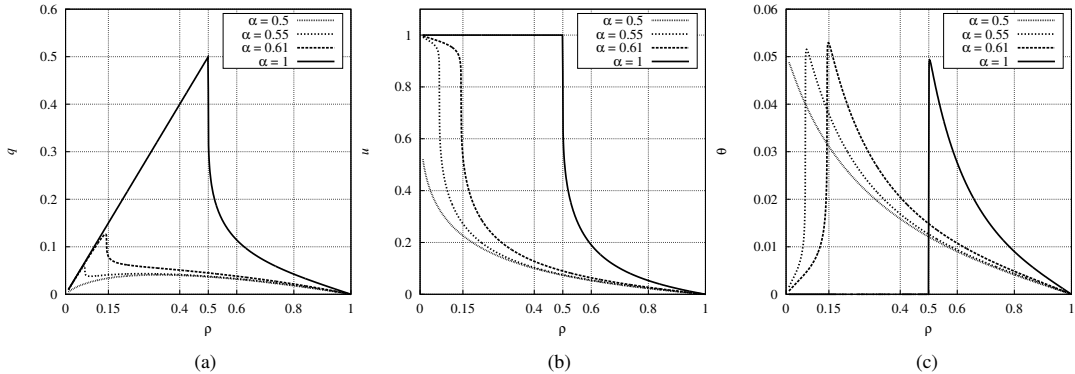}
\caption{From left to right, fundamental diagrams of the average flux, the average speed, and the speed variance computed from the stationary solutions of the spatially homogeneous model~\eqref{eq:spat.homog} for various values of the parameter $\alpha$ indicated in the key.}
\label{fig:fund.diag}
\end{figure}

Diagrams in Figs.~\ref{fig:fund.diag}a, \ref{fig:fund.diag}b have been computed by duly adapting to the present case the formulas reported in Eq.~\eqref{eq:macro.var}. They demonstrate that the model is able to catch qualitatively the well-known \emph{phase transition} from free to congested traffic flow \cite{Kerner2004pot}. The trend of the average flux $q$ (Fig.~\ref{fig:fund.diag}a) is basically linear for low density, then it becomes markedly nonlinear when $\rho$ increases, i.e., the road progressively clogs up. Notice that for $\alpha=1$, i.e., in the ideal situation of optimal environmental conditions, the average flux is exactly linear in the free flow phase, which persists until $\rho=0.5$. The average speed $u$ (Fig.~\ref{fig:fund.diag}b) is initially almost constant and close to the maximum possible value $u=1$ (actually, for $\alpha=1$ the model predicts that it is exactly constant to $1$ until $\rho=0.5$), then it drops steeply to zero when the density enters the congested flow range. The free flow phase reduces as the environmental conditions worsen. For $\alpha$ decreasing from $1$ to $0$, it eventually disappears completely (cf. the curves for $\alpha=0.5$).

The curves in Fig.~\ref{fig:fund.diag}c show the asymptotic trend of the speed variance $\theta$, defined as
\begin{equation*}
	\theta:=\frac{1}{\rho}\sum_{j=1}^{n}{(v_j-u)}^2f_j,
\end{equation*}
vs. the car density $\rho$. They can provide useful quantitative information about the aforesaid phase transition, considering that they clearly show that there is a value of $\rho$ (depending on $\alpha$) for which $\theta$ is maximum. By understanding such a maximum as the signal of a strong change in the microscopic behavior of vehicles with respect to the estimated average one, it is possible to interpret the corresponding density as the \emph{critical density} $\rho_c$ separating the free and congested regimes. For instance, experimental measurements on the Venice-Mestre highway \cite{bonzani2003eht} indicate that a realistic range of the critical density is $\rho_c\in[0,\,0.15]$. According to the presented simulations, the corresponding realistic range of the environmental parameter is $\alpha\in[0,\,0.61]$. Of course, different roads can produce different empirical evaluations of $\rho_c$, however the tuning procedure of $\alpha$ based on the fundamental diagram of $\theta$ remains conceptually the same.

A couple of comments are in order:
\begin{itemize}
\item the simulations based on Eq.~\eqref{eq:spat.homog} indicate that for $\alpha\leq 0.5$ the speed variance $\theta$ attains its maximum always at $\rho=0$, i.e., the free flow phase is absent. Therefore the trend of the variance does not provide enough information for identifying univocally the value of $\alpha$ corresponding to a null critical density, this particular case requiring presumably a more accurate investigation;
\item the same simulations indicate also that in no case the critical density is expected to be greater than $\rho_c=0.5$ (which is the value found for $\alpha=1$). This can be explained considering that for $\rho>0.5$ cars are, in average, closer to one another than their characteristic size (i.e., the average distance between two successive cars is lower than the typical car size), which might prevent any of them from being unperturbed by cars ahead.
\end{itemize}

\begin{table}[!t]
\centering
\begin{tabular}{lcl}
Parameter & Value & Description \\
\hline
$m$ & $10$ & Number of space cells \\
$n$ & $6$ & Number of speed classes \\
$\eta_0$ & $1$ & Coefficient of the interaction rate, see Eq.~\eqref{eq:eta} \\
$\alpha$ & $0.61$ & Constant environmental conditions \\
$\alpha_i$ & {\scriptsize $\begin{cases} 0.61 & \text{if\ } i\leq 5 \\ \frac{1}{40}\left(31-\frac{i}{10}\right) & \text{if\ } 6\leq i\leq 9 \\ 0.5 & \text{if\ } i=10\end{cases}$}
	& Cell-to-cell variable environmental conditions \\
$\beta$ & $0$ & Coefficient of the fictitious density, see Eq.~\eqref{eq:fictrho} \\
\hline
\end{tabular}
\caption{Model parameters for the roadworks problem}
\label{tab:roadworks}
\end{table}

\begin{table}[!t]
\centering
\begin{tabular}{lcl}
Parameter & Value & Description \\
\hline
$m$ & $10$ & Number of space cells \\
$n$ & $6$ & Number of speed classes \\
$\eta_0$ & $1$ & Coefficient of the interaction rate, see Eq.~\eqref{eq:eta} \\
$\alpha$ & $0.55$ & Environmental conditions \\
$\beta$ & $1$ & Coefficient of the fictitious density, see Eq.~\eqref{eq:fictrho} \\
\hline
\end{tabular}
\caption{Model parameters for the traffic light problem}
\label{tab:traffic.light}
\end{table}

\subsection{The spatially inhomogeneous problem}
A second case study concerns the evolution in time of the car density along a road. Nonuniform space conditions are now explicitly considered, the interest being also in the transient trend of the traffic flow. A qualitative validation of the model can be performed upon assessing its ability to reproduce spatial flow patterns observed in real conditions. Here we focus specifically on the formation of queues, presenting two prototypical scenarios which may be of practical interest for traffic management. All relevant parameters used in the next simulations are summarized in Tables~\ref{tab:roadworks}, \ref{tab:traffic.light}.

The first scenario refers to queue formation caused by a worsening of the driving conditions, for instance a narrowing of the roadway due to roadworks. Using the full Eq.~\eqref{eq:balance.fij.complete}, this can be simulated by assuming that the parameter $\alpha$ in the table of games varies from cell to cell, $\alpha=\alpha_i$, $i=1,\,\dots,\,m$. In particular, inspired by the reasonings proposed for the spatially homogeneous problem, we consider a gradual decrease of $\alpha$ from $0.61$ at the beginning of the road to $0.5$ at its end. In the case of the already mentioned Venice-Mestre highway, this corresponds to a change in the expected flow conditions from $\rho_c=0.15$ to $\rho_c=0$. Figure~\ref{fig:roadworks} clearly shows the formation and backward propagation of a queue for such a variable $\alpha$ with respect to the basically unperturbed flow of vehicles produced when $\alpha$ is kept constant to its best possible value ($\alpha=0.61$) all along the road.

\begin{figure}[!t]
\centering
\includegraphics[width=0.95\textwidth]{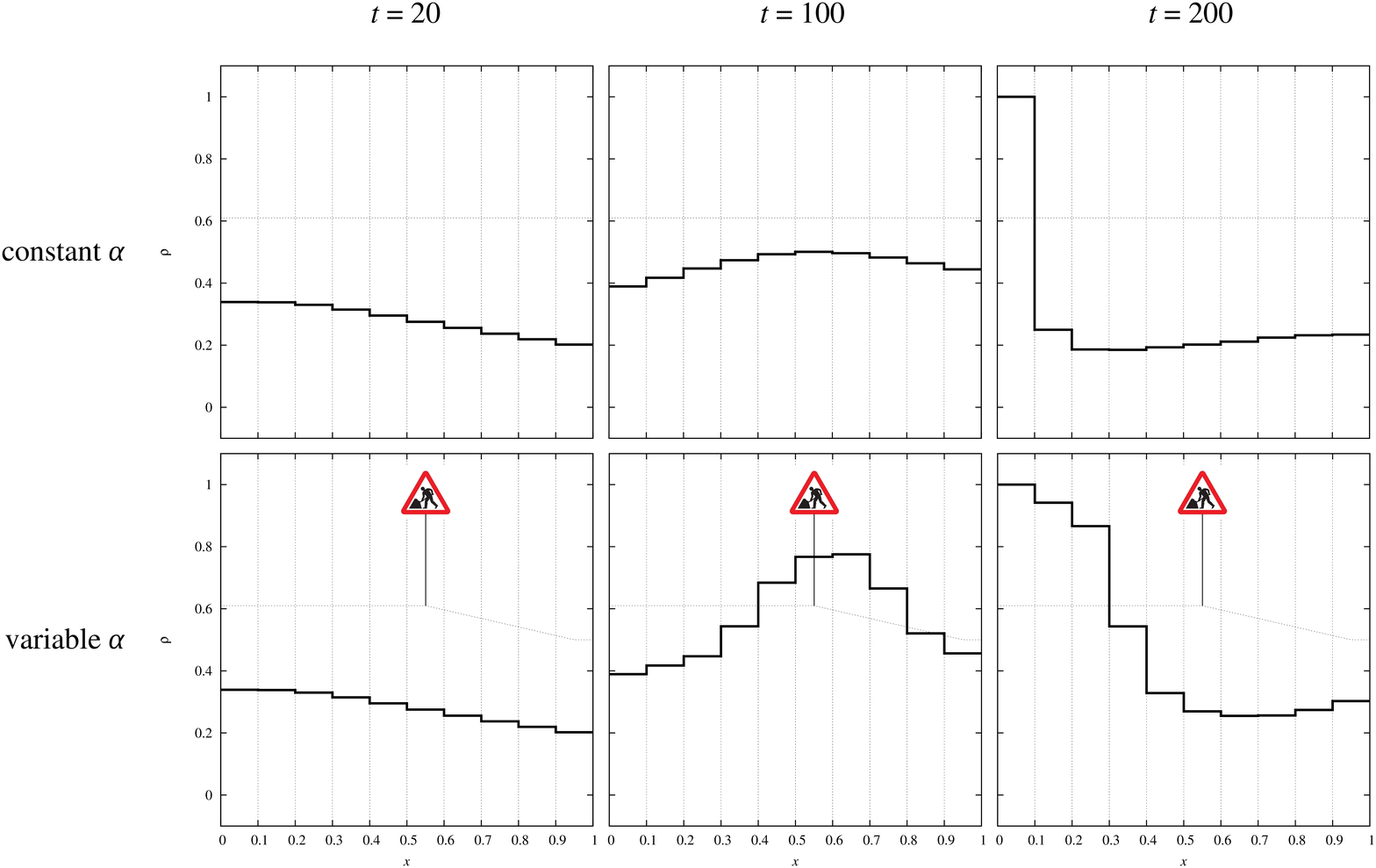}
\caption{Comparison, at the same computational times, of the traffic flow under constant and variable road conditions due to roadworks. The road is initially empty, while a constant density $\rho_0$ of incoming vehicles, uniformly distributed over the speed classes, is imposed at the left boundary at all times. The flux limiters $\Phi_{0,1}$, $\Phi_{10,11}$ at the left and right boundaries are set to $\frac{1-\rho_1}{\rho_0}$ and to $1$, respectively. The dotted line spanning the space cells is the graph of the function $\alpha=\alpha(x)$ used in these simulations.}
\label{fig:roadworks}
\end{figure}
\begin{figure}[!t]
\centering
\includegraphics[width=0.95\textwidth]{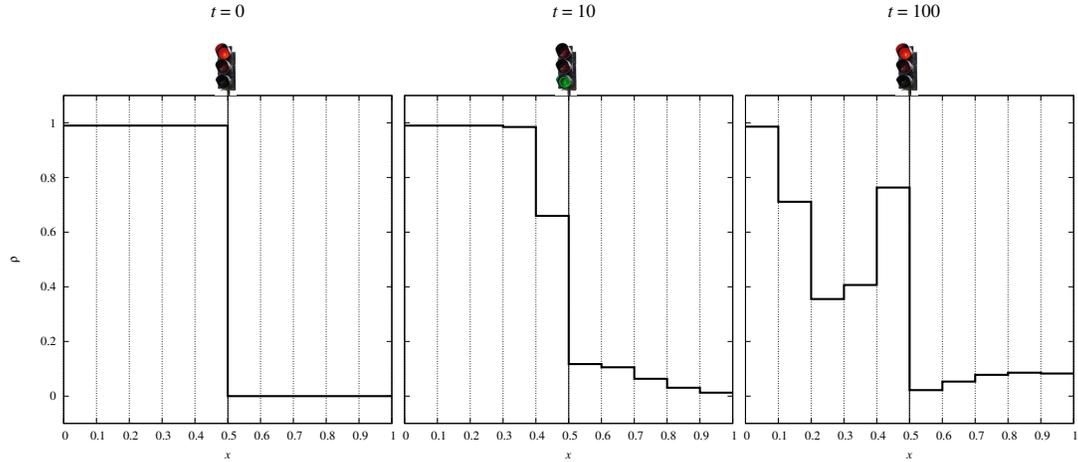}
\caption{Evolution of a queue during several red-green cycles of a traffic light placed at $x=0.5$. No cars enter the domain from the left boundary for all $t>0$. The initial condition is the pre-formed queue illustrated in the left snapshot. The traffic light period is of $20$ computational time units.}
\label{fig:traffic.light}
\end{figure}

The second scenario refers instead to the alternate formation and depletion of a queue behind a traffic light, which, as Fig.~\ref{fig:traffic.light} demonstrates, are well reproduced by the model. In more detail, the action of the traffic light on the flow of vehicles is easily simulated, in our framework, by identifying the traffic light with one of the flux limiters ($\Phi_{5,6}$ in the case of Fig.~\ref{fig:traffic.light}), to which a periodic time evolution is assigned mimicking the successive red-green cycles of the light. Specifically, during the red phase the selected limiter is set invariably to zero, since no car is allowed to flow beyond the traffic light regardless of the actual free space in the cell ahead. Conversely, during the green phase the selected limiter switches to the expression given by Eq.~\eqref{eq:flux.lim}, hence the flow of cars is again possible according to the usual rules. The initial condition for this problem is, as shown by the first snapshot of Fig.~\ref{fig:traffic.light}, a fully formed queue (i.e., $\rho_i=1$) standing at a red light with empty road ahead. When, at $t>0$, the light turns to green, cars in the cell $I_5=[0.4,\,0.5]$ can start moving only if they are able to see the empty space in front of them, namely if they react to a fictitious density $\fictrho_5$ lower than the real one $\rho_5=1$. Indeed, by carefully inspecting the table of games reported in Eqs.~\eqref{eq:tab.games.h<k}--\eqref{eq:tab.games.h=k}, it results that standing vehicles (i.e., vehicles in the speed class $v_1=0$) in a full cell (i.e., one with $\rho_i=1$) have no chance to shift to the velocity class $v_2>0$ if the transition probabilities depend on the actual cell density. Therefore, the concept of fictitious density as introduced in Eq.~\eqref{eq:fictrho} is important for conferring realism to the model.

\section{Qualitative analysis}
\label{sect:qual.anal}
The aim of this section is to study the well-posedness of the following class of initial/boundary-value problems, generated by the mathematical framework presented in the previous sections:

\begin{equation}
\begin{cases}
	\dfrac{df_{ij}}{dt}+v_j\left(\Phi_{i,i+1}[\f]f_{ij}-\Phi_{i-1,i}[\f]f_{i-1,j}\right)
		=\eta(i)[\f]\left(\displaystyle{\sum_{h,\,k=1}^{n}}A_{hk}^{j}(i)[\f]f_{ih}f_{ik}
			-f_{ij}\rho_i\right) \\
	f_{ij}(0)=f^0_{ij} \\
	f_{0j}(t)=\bar{f}_j(t) \\
	\Phi_{0,1}=\bar{\Phi}(t),
\end{cases}
\label{IBVP}
\end{equation}
which we specifically consider in the mild formulation
\begin{equation}
	\begin{cases}
		\begin{aligned}
			f_{ij}(t)={f}^0_{ij}+\displaystyle{\int_0^t}\{ &v_j \left(\Phi_{i-1,i}[\f](s)f_{i-1,j}(s)
				-\Phi_{i,i+1}[\f](s)f_{ij}(s)\right) \\
			& -G_{ij}[\f,\,\f](s)+f_{ij}(s)L_{ij}[\f](s)\}\,ds
		\end{aligned} \\[5mm]
		f_{0j}(t)={\bar{f}}_j(t) \\
		\Phi_{0,1}=\bar{\Phi}(t)
	\end{cases}
	\label{IBVPmild}
\end{equation}
where
\begin{equation}
	G_{ij}[\f,\,\f]=\eta(i)[\f]\sum_{h,\,k=1}^{n}A_{hk}^{j}(i)[\f]f_{ih}f_{ik}, \qquad
		L_{ij}[\f]=\eta(i)[\f]\rho_i
	\label{dis_op}
\end{equation}
are the so-called \emph{gain} and \emph{loss operator}, respectively (cf. Eq.~\eqref{eq:J}). Notice that unlike the previous sections here we have emphasized the dependence of the flux limiters, the interaction rate, and the table of games on the set of distribution functions $\f=\{f_{ij}\}_{i,j}$. This notation will be indeed useful in the next calculations.

We organize the qualitative analysis of Problem~\eqref{IBVP} in two parts: first we prove the uniqueness of the solution and its continuous dependence on the initial and boundary data; next we state existence. As a preliminary to this, we introduce the functional framework in which we will consider Problem~\eqref{IBVP} and the assumptions needed for achieving the proofs.

We denote by $X_{T_\M}=C([0,\,T_\M];\,\R^{mn})$ the Banach space of vector-valued continuous functions $\u=\u(t):[0,\,T_\M]\to\R^{mn}$, $\u(t)=(u_{11}(t),\,\dots,\,u_{1n}(t),\,\dots,\,u_{m1}(t),\,\dots,\,u_{mn}(t))$, equipped with the uniform norm
$$ \norm{\u}_\infty:=\sup_{t\in[0,\,T_\M]}\norm{\u(t)}_1
	=\sup_{t\in[0,\,T_\M]}\sum_{i=1}^m\sum_{j=1}^n\abs{u_{ij}(t)}, $$
where $\norm{\cdot}_1$ is the classical 1-norm in $\R^{mn}$. We also introduce its subset
$$ \B=\left\{\u\in X_{T_\M}:\, 0\leq u_{ij}(t)\leq 1,\ 
	\sum_{j=1}^n u_{ij}(t)\leq 1,\ \forall\,i,\,j,\ \forall\,t\in[0,\,T_\M]\right\}, $$
which is readily seen to be closed. We point out that, by virtue of the equivalence of all norms in finite dimension, a different $p$-norm may be considered in $\R^{mn}$. However, the 1-norm takes advantage more directly of the interpretation in terms of number density: if $\u\in\B$ then $\norm{\u(t)}_1$ is the density of the distribution $\u$ at time $t$ (cf. Eq.~\eqref{eq:N}).

A solution to our initial/boundary-value problem is defined as follows:
\begin{definition}
A function $\f=\f(t):[0,\,T_\M]\to\R^{mn}$ is said to be a \emph{mild solution} to Problem~\eqref{IBVP} if $\f\in\B$ and $\f$ satisfies~\eqref{IBVPmild}.
\label{def:mild.solution}
\end{definition}
Notice that $\f\in\B$ guarantees that $\f$ fulfills the basic requirements of a physically acceptable distribution function: non-negativity and boundedness in each space/speed class in the whole time interval of existence.

The following box summarizes the hypotheses, to be possibly regarded as modeling guidelines, under which we will achieve the proofs of the announced results.

\begin{framed}
\centering\textbf{Assumptions for Problem~\eqref{IBVP}}
\begin{enumerate}
\item \label{ip1} The initial and boundary data $f^0_{ij}$, $\bar{f}_{j}$ belong to $\B$.
\item \label{ip2} The flux limiter $\Phi_{i,i+1}[\f]$ is defined by a function $\Phi:[0,\,1] \times [0, \,1] \to [0,\,1]$ acting on the densities of the two adjacent cells $I_1$, $I_{i+1}$ such that $\Phi_{i,i+1}[\f]=\Phi(\rho_i,\,\rho_{i+1})$. The function $\Phi$ has the following properties:
\begin{enumerate}
\item[(i)] $0\leq\Phi(u,\,v)\leq 1$, $\forall\,u,\,v\in [0,\,1]$;
\item[(ii)] $\Phi(u,\,v)u\leq 1-v$, $\forall\,(u,\,v)\in [0,\,1]\times[0,\,1]$ such that $u+v>1$;
\item[(iii)] $\abs{\Phi(u_2,\,v_2)-\Phi(u_1,\,v_1)}\leq\Lip(\Phi)(\abs{u_2-u_1}+\abs{v_2-v_1})$, $\forall\,(u_1,\,v_1),\,(u_2,v_2)\in [0,\,1]\times [0,\,1]$.
\end{enumerate}
\item \label{ip3} The boundary datum $\bar{\Phi}$ is as in Assumption \ref{ip2}, i.e., $\bar{\Phi}=\Phi(\bar{\rho},\,\rho_1)$ where $\bar{\rho}=\sum_{j=1}^{n}\bar{f}_j$ is the prescribed density of incoming cars.
\item The interaction rate $\eta(i)$ is:
\begin{enumerate}
\item \label{ip4a} non-negative and uniformly bounded, i.e., there exists a constant $\bar{\eta}>0$ such that
$$ 0\leq\eta(i)[\u]\leq\bar{\eta}, \quad \forall\,i=1,\,\dots\,m,\ \forall\,\u\in\B; $$
\item \label{ip4b} Lipschitz continuous with respect to the distribution in $\B$, i.e.,
$$ \abs{\eta(i)[\u](t)-\eta(i)[\v](t)}\leq\Lip(\eta(i))\norm{\u(t)-\v(t)}_1,
	\quad \forall\,t\in [0,\, T_\M],\ \forall\,\u,\,\v\in\B. $$
\end{enumerate}
\item \label{ip5} The elements $A^j_{hk}(i)$ of the table of games satisfy Eq.~\eqref{eq:tog.sum.1} and the following Lipschitz continuity condition:
$$ \abs{A^j_{hk}(i)[\u](t)-A^j_{hk}(i)[\v](t)}\leq \Lip(A^j_{hk}(i))\norm{\u(t)-\v(t)}_1,
	\quad \forall\,t\in [0,\,T_\M],\ \forall\,\u,\,\v\in\B. $$
\end{enumerate}	
\end{framed}

We henceforth assume that Assumptions~\ref{ip1}--\ref{ip4a} are always satisfied, indeed they are needed in all of the next results. On the contrary, we will explicitly mention Assumptions~\ref{ip4b}, \ref{ip5} only when they will be actually used.

\subsection{Uniqueness and continuous dependence on the data}
The following theorem states that the mild solution to Problem~\eqref{IBVP} is unique in $\B$ and that it depends continuously on the initial and boundary data.

\begin{theorem}[Uniqueness and continuous dependence]
Under Assumptions~\ref{ip4b}, \ref{ip5}, let $(\f^0,\,\bar{\f},\,\bar{\Phi}^f)$, $(\g^0,\,\bar{\g},\,\bar{\Phi}^g)$ be two sets of initial and boundary data and $\f,\,\g\in\B$ two corresponding mild solutions to Problem~\eqref{IBVP}. Then there exists $\C>0$ such that:
$$ \norm{\f-\g}_\infty\leq\C\left[\norm{\f^0-\g^0}_1+\int_0^{T_\M}\left(\abs{\bar{\Phi}^f(t)-\bar{\Phi}^g(t)}
	+\norm{\bar{\f}(t)-\bar{\g}(t)}_1\right)\,dt\right]. $$
In particular, there is at most one solution corresponding to a given set of initial and boundary data.
\label{theorem1}
\end{theorem}
\begin{proof}
Subtracting term by term the mild equations satisfied by $\f$, $\g$ (cf. Eq.~\eqref{IBVPmild}) and then summing over $i=1,\,\dots,\,m$ and $j=1,\,\dots,\,n$ gives:
\begin{align*}
	\norm{\f(t)-\g(t)}_1\leq\norm{\f^0-\g^0}_1 &+ \sum_{i=1}^m\sum_{j=1}^n\Biggl[\int_0^t
		\vert v_j(\Phi_{i-1,i}[\f](s) f_{i-1,j}(s)-\Phi_{i-1,i}[\g](s)g_{i-1,j}(s) \\
	&\phantom{+} \qquad\qquad +\Phi_{i,i+1}[\g](s)g_{ij}(s)-\Phi_{i,i+1}[\f](s)f_{ij}(s))\vert\,ds \\
	&\phantom{+} +\int_0^t\Bigl(\abs{G_{ij}[\g,\,\g](s)-G_{ij}[\f,\,\f](s)}+
		\vert f_{ij}(s)L_{ij}[\f](s) \\
	&\phantom{+} \qquad\qquad -g_{ij}(s)L_{ij}[\g](s)\vert\Bigr)\,ds\Biggr].
\end{align*}
By exploiting the Lipschitz continuity of the various terms and extracting from the sums the boundary terms corresponding to $i=1$ we then discover:
\begin{multline}
	\sum_{i=1}^m\sum_{j=1}^n\abs{\Phi_{i-1,i}[\f](s)f_{i-1,j}(s)-\Phi_{i-1,i}[\g](s)g_{i-1,j}(s)
		+\Phi_{i,i+1}[\g](s)g_{ij}(s)-\Phi_{i,i+1}[\f](s)f_{ij}(s)} \\
			\leq 2(2\Lip(\Phi)+1)\norm{\f(s)-\g(s)}_1+\abs{\bar{\Phi}^f(s)-\bar{\Phi}^g(s)}
				+\norm{\bar{\f}(s)-\bar{\g}(s)}_1,
\label{1term}
\end{multline}
\begin{multline}
	\sum_{i=1}^{m}\sum_{j=1}^{n}\abs{G_{ij}[\g,\,\g](s)-G_{ij}[\f,\,\f](s)} \\
		\leq\left[2\bar{\eta}n+\sum_{i=1}^m\left(n^2\Lip(\eta(i))+
			\bar{\eta}\sum_{h,\,k,\,j=1}^n\Lip(A^j_{hk}(i))\right)\right]\norm{\f(s)-\g(s)}_1,	
	\label{2term}
\end{multline}
\begin{equation}
	\sum_{i=1}^{m}\sum_{j=1}^{n}\abs{f_{ij}(s)L_{ij}[\f](s)-g_{ij}(s)L_{ij}[\g](s)}
		\leq\left(2\bar{\eta}n+n^2\sum_{i=1}^m\Lip(\eta(i))\right)\norm{\f(s)-\g(s)}_1,
	\label{3term}
\end{equation}
whence
$$ \norm{\f(t)-\g(t)}_1\leq\norm{\f^0-\g^0}_1+\C\int_0^t\left(\norm{\f(s)-\g(s)}_1+
		\abs{\bar{\Phi}^f(s)-\bar{\Phi}^g(s)}+\norm{\bar{\f}(s)-\bar{\g}(s)}_1\right)\,ds $$
where $\C$ is a suitable constant deduced from all those appearing in Eqs.~\eqref{1term}--\eqref{3term}. Gronwall's inequality implies now:
$$ \norm{\f(t)-\g(t)}_1\leq e^{\C t}\left[\norm{\f^0-\g^0}_1+
	\int_0^t\left(\abs{\bar{\Phi}^f(s)-\bar{\Phi}^g(s)}+\norm{\bar{\f}(s)-\bar{\g}(s)}_1\right)\,ds\right] $$
whence, taking the supremum over $t\in[0,\,T_\M]$ on both sides, we obtain the continuous dependence estimate. Uniqueness follows straightforwardly from it by letting $\f^0=\g^0$, $\bar{\f}=\bar{\g}$, and $\bar{\Phi}^f=\bar{\Phi}^g$.
\end{proof}

\subsection{Existence}
In this section we prove the existence of the mild solution to Problem~\eqref{IBVP}. The idea of the proof is  to consider at first the model at discrete time instants $t^\kk_\n=\n\Delta{t}_\kk$, where the index $\n$ labels the discrete time and the index $\kk$ is a mesh parameter denoting the level of refinement of the time grid. In particular, the time step $\Delta{t}_\kk$ is chosen such that it tends to zero for $\kk\to\infty$. After establishing some technical properties of the iterates of such a discrete-in-time model (cf. Lemma~\ref{lemma1}), we pass from discrete to continuous time by interpolating the iterates and we prove that in the limit $\kk\to\infty$ this interpolation converges to a distribution function $\f\in\B$ (cf. Lemma~\ref{lemma2}). Finally, we show that $\f$ is precisely a mild solution of to Problem~\eqref{IBVP} (cf. Theorem~\ref{theorem2}).

\paragraph*{Step 1}
We consider Problem~\eqref{IBVP} at discrete time instants, cf. Eq~\eqref{dis.eq.balance}:
\begin{align}
	f^{\n+1,\kk}_{ij}=f^{\n,\kk}_{ij} &- \Delta{t}_\kk v_j(\Phi_{i,i+1}[\f^{\n,\kk}]f^{\n,\kk}_{ij}-
		\Phi_{i-1,i}[\f^{\n,\kk}]f^{\n,\kk}_{i-1,j}) \nonumber \\
	&+ \Delta{t}_\kk(G_{ij}[\f^{\n,\kk},\,\f^{\n,\kk}]-f^{\n,\kk}_{ij}L_{ij}[\f^{\n,\kk}]).
	\label{discrete}
\end{align}
The following lemma gives some properties of the iterates $f^{\n,\kk}_{ij}$.

\begin{lemma}
If $\Delta{t}_\kk$ is sufficiently small then $\f^{\n,\kk}\in\B$ for all $\n,\,\kk\geq 0$.
\label{lemma1}
\end{lemma}
\begin{proof}
The proof is by induction on $\n$, assuming that $\f^{\n,\kk}\in\B$ and that the initial and boundary data have been chosen consistently with the properties of the distribution functions in $\B$.

\begin{enumerate}
\item[(i)] \emph{Non-negativity of the iterates}. Taking into account the non-negativity of the terms containing $\Phi_{i-1,i}[\f^{\n,\kk}]$ and $G_{ij}[\f^{\n,\kk},\,\f^{\n,\kk}]$, it results:
$$ f^{\n+1,\kk}_{ij}\geq f^{\n,\kk}_{ij}[1-\Delta{t}_\kk(v_j\Phi_{i,i+1}[\f^{\n,\kk}]+L_{ij}[\f^{\n,\kk}])]
	\geq f^{\n,\kk}_{ij}[1-\Delta{t}_\kk(1+L_{ij}[\f^{\n,\kk}])], $$
being $v_j,\,\Phi_{i,i+1}[\f^{\n,\kk}]\leq 1$. Moreover $L_{ij}[\f^{\n,\kk}]\leq\bar{\eta}\rho_i^{\n,\kk}\leq\bar{\eta}$, hence if $\Delta{t}_\kk<1/(1+\bar{\eta})$ we deduce
$$ f^{\n+1,\kk}_{ij}\geq f^{\n,\kk}_{ij}[1-\Delta{t}_\kk(1+\bar{\eta})]\geq 0. $$
\item[(ii)] \emph{Boundedness of the iterates}. Considering that $v_j\leq 1$ we have:
$$ f^{\n+1,\kk}_{ij}\leq f^{\n,\kk}_{ij}+\Delta{t}_\kk(\Phi_{i-1,i}[\f^{\n,\kk}]f^{\n,\kk}_{i-1,j}
	+ G_{ij}[\f^{\n,\kk},\,\f^{\n,\kk}]-f^{\n,\kk}_{ij}L_{ij}[\f^{\n,\kk}]). $$
But
$$ G_{ij}[\f^{\n,\kk},\,\f^{\n,\kk}]-f^{\n,\kk}_{ij}L_{ij}[\f^{\n,\kk}]
	\leq\bar{\eta}\left[1-\left(f^{\n,\kk}_{ij}\right)^2\right], $$
thus we conclude
\begin{equation}
	f^{\n+1,\kk}_{ij}\leq f^{\n,\kk}_{ij}+\Delta{t}_\kk\left\{\Phi_{i-1,i}[\f^{\n,\kk}]f^{\n,\kk}_{i-1,j}
		+\bar{\eta}\left[1-\left(f^{\n,\kk}_{ij}\right)^2\right]\right\}.
	\label{f:magg}
\end{equation}
Now, if $f^{\n,\kk}_{ij}=1$ then from Assumption~\ref{ip2}(ii) it follows $f^{\n+1,\kk}_{ij}\leq 1$ as desired. Conversely, if $f^{\n,\kk}_{ij}<1$ then according to Assumption~\ref{ip2}(i)-(ii)
it results
\begin{equation}
	\Phi_{i-1,i}[\f^{\n,\kk}]f^{\n,\kk}_{i-1,j}\leq
		\Phi_{i-1,i}[\f^{\n,\kk}]\rho^{\n,\kk}_{i-1}\leq
			\left\{
			\begin{array}{ll}
				1-\rho^{\n,\kk}_i & \text{if\ } \rho^{\n,\kk}_i+\rho^{\n,\kk}_{i-1}>1 \\
				\rho^{\n,\kk}_{i-1}\leq 1-\rho^{\n,\kk}_i & \text{if\ } \rho^{\n,\kk}_i+\rho^{\n,\kk}_{i-1}\leq 1
			\end{array}
			\right\}
			\leq 1-\rho^{\n,\kk}_i.
	\label{phi_rho}
\end{equation}
Hence we can continue the estimate~\eqref{f:magg} as:
$$ f^{\n+1,\kk}_{ij}\leq f^{\n,\kk}_{ij}+\Delta{t}_\kk\left\{1-\rho^{\n,\kk}_{i}
	+\bar{\eta}\left[1-\left(f^{\n,\kk}_{ij}\right)^2\right]\right\}
		\leq f^{\n,\kk}_{ij}+\Delta{t}_\kk(1+2\bar{\eta})(1-f^{\n,\kk}_{ij}) $$
whence we see again that $f^{\n+1,\kk}_{ij}\leq 1$ if $\Delta{t}_\kk<1/(1+2\bar{\eta})$.
\item[(iii)] \emph{Boundedness of the sum of the iterates}. In order to prove that $\sum_{j=1}^nf^{\n,\kk}_{ij}\leq 1$, we note that from Eq.~\eqref{eq:tog.sum.1} it follows
$\sum_{j=1}^n(G_{ij}[\f^{\n,\kk},\,\f^{\n,\kk}]-f^{\n,\kk}_{ij}L_{ij}[\f^{\n,\kk}])=0$. Consequently, by Eq.~\eqref{phi_rho} we obtain the thesis, since
$$ \sum_{j=1}^n f^{\n+1,\kk}_{ij}\leq \rho^{\n,\kk}_i
	+\Delta{t}_\kk\Phi_{i-1,i}[\f^{\n,\kk}]\rho^{\n,\kk}_{i-1}
		\leq\rho^{\n,\kk}_i+\Delta{t}_\kk(1-\rho^{\n,\kk}_i)\leq 1 $$
provided $\Delta{t}_\kk\leq 1$.
\end{enumerate}

Finally, we note that all properties proved so far hold simultaneously if $\Delta{t}_\kk<1/(1+2\bar{\eta})$.
\end{proof}

\paragraph*{Step 2}
Now we pass by interpolation from discrete to continuous time. To this end, we introduce the following function $\hat{\f}=\hat{\f}^\kk(t):[0,\,T_\M]\to\R^{mn}$ interpolating piecewise linearly the iterates $f^{\n,\kk}_{ij}$:
\begin{equation}
	\hat{f}^\kk_{ij}(t)=\sum_{\n=1}^{N_\kk}\left[\left(1-\frac{t-t^\kk_{\n-1}}{\Delta{t}_\kk}\right)f^{\n-1,\kk}_{ij}
		+\frac{t-t^\kk_{\n-1}}{\Delta{t}_\kk}f^{\n,\kk}_{ij}\right]\1_{[t^\kk_{\n-1},\,t^\kk_\n]}(t).
	\label{interpolante}
\end{equation}
Here $N_\kk$ is the total number of time steps in the $\kk$-th mesh. We assume that $N_\kk$ and $\Delta{t}_\kk$ are chosen in such a way that $N_\kk\Delta{t}_\kk=T_\M$ independently of the refinement parameter $\kk$.

For the function $\hat{\f}^\kk$ we can prove the following result.
\begin{lemma}
There exists  $\f \in \B$ such that
$$ \lim_{\kk\to\infty}\norm{\hat{\f}^\kk-\f}_\infty=0. $$
\label{lemma2}
\end{lemma}
\begin{proof}
We first claim that $\hat{\f}^\kk\in\B$. In fact, for all $t\in[0,\,T_\M]$ there exists $\bar{\n}\leq N_\kk$ such that $t\in [t^\kk_{\bar{\n}-1},\,t^\kk_{\bar{\n}}]$, thus:
$$ \hat{f}^\kk_{ij}(t)=\left(1-\frac{t-t^\kk_{\bar{\n}-1}}{\Delta{t}_\kk}\right)f^{\bar{\n}-1,\kk}_{ij}
	+\frac{t-t^\kk_{\bar{\n}-1}}{\Delta{t}_\kk}f^{\bar{\n},\kk}_{ij} $$
and according to Lemma~\ref{lemma1} this value satisfies all constraints required pointwise in time to the components of the functions in $\B$. Thus the claim follows from the arbitrariness of $t$.

Next, our goal is to show that the sequence $\{\hat{\f}^\kk\}_{\kk\geq 0}$ is relatively compact in $X_{T_\M}$. Indeed in such a case we can conclude that (up to subsequences) it converges uniformly to some $\f\in X_{T_\M}$. Moreover, since $\{\hat{\f}^\kk\}_{\kk\geq 0}\subset\B$ and $\B$ is closed, we also have that $\f\in\B$. We prove the relative compactness by means of the Ascoli-Arzel\`a compactness criterion:
\begin{enumerate}
\item[(i)] $\{\hat{\f}^\kk\}_{\kk\geq 0}$ is uniformly bounded. Indeed, since $\hat{\f}^\kk\in\B$ we can write
$$ \sup_{\kk\geq 0}\norm{\hat{\f}^\kk}_\infty=\sup_{\kk\geq 0}\sup_{t\in [0,\,T_\M]}
	\sum_{i=1}^m\sum_{j=1}^n\abs{\hat{f}^\kk_{ij}(t)}\leq m<\infty. $$
\item[(ii)] $\{\hat{\f}^\kk\}_{\kk\geq 0}$ is equicontinuous. In fact, let $t_1\in [t^\kk_{\bar{\m}-1},\,t^\kk_{\bar{\m}}]$ and $t_2\in [t^\kk_{\bar{\n}-1},\,t^\kk_{\bar{\n}}]$ with $\m\leq\n$ (i.e., $t_1\leq t_2$), then we can write:
\begin{align}
	\hat{f}^\kk_{ij}(t_2)-\hat{f}^\kk_{ij}(t_1) &= \hat{f}^\kk_{ij}(t_2)+
		\sum_{\ell=\bar{\m}}^{\bar{\n}-1}\hat{f}^\kk_{ij}(t^\kk_\ell)-
			\sum_{\ell=\bar{\m}}^{\bar{\n}-1}\hat{f}^\kk_{ij}(t^\kk_\ell)
				-\hat{f}^\kk_{ij}(t_1) \nonumber \\
	&= [\hat{f}^\kk_{ij}(t_2)-\hat{f}^\kk_{ij}(t^\kk_{\bar{\n}-1})]+
		[\hat{f}^\kk_{ij}(t^\kk_{\bar{\m}})-\hat{f}^\kk_{ij}(t_1)]+
			\sum_{\ell=\bar{\m}}^{\bar{\n}-2}(f^{\ell+1,\kk}_{ij}-f^{\ell,\kk}_{ij}) \nonumber \\
	&=: A_1+A_2+A_3.
\label{equicon}
\end{align}

We begin by estimating the term $A_3$. By Eq.~\eqref{discrete} and applying the same arguments as in the proof of Lemma~\ref{lemma1} we have:
\begin{align*}
	\abs{f^{\ell+1,\kk}_{ij}-f^{\ell,\kk}_{ij}} &= \Delta{t}_\kk\abs{v_j(\Phi_{i-1,i}[\f^{\ell,\kk}]f^{\ell,\kk}_{i-1,j}
		-\Phi_{i,i+1}[\f^{\ell,\kk}]f^{\ell,\kk}_{i,j})+G_{ij}[\f^{\ell,\kk},\,\f^{\ell,\kk}]
			-f^{\ell,\kk}_{ij}L_{ij}[\f^{\ell,\kk}]} \\
	& \leq 2(1+\bar{\eta})\Delta{t}_\kk,
\end{align*}
thus
\begin{equation}
	\abs{A_3}\leq\sum_{\ell=\bar{\m}}^{\bar{\n}-2}\abs{f^{\ell+1,\kk}_{ij}-f^{\ell,\kk}_{ij}}
		\leq 2(\bar{\n}-1-\bar{\m})(1+\bar{\eta})\Delta{t}_\kk=2(1+\bar{\eta})(t^\kk_{\bar{\n}-1}-t^\kk_{\bar{\m}}).
	\label{A_3}
\end{equation}

As for the term $A_2$, we observe that taking into account Eqs.~\eqref{interpolante}, \eqref{A_3} we can write
\begin{align}
	\abs{A_2}=\abs{\hat{f}^\kk_{ij}(t_2)-\hat{f}^\kk_{ij}(t^\kk_{\bar{\n}-1})} &=
 		\abs{\hat{f}^\kk_{ij}(t_2)-f^{\bar{\n}-1,\kk}_{ij}} \nonumber \\
 	&= \frac{t_2-t^\kk_{\bar{\n}-1}}{\Delta{t}_\kk}\abs{f^{\bar{\n},\kk}_{ij}-f^{\bar{\n}-1,\kk}_{ij}}
 		\leq 2(1+\bar{\eta})(t_2-t^\kk_{\bar{\n}-1}).
	\label{A_2}
\end{align}

Likewise,
\begin{align}
	\abs{A_1}=\abs{\hat{f}^\kk_{ij}(t^\kk_{\bar{\m}})-\hat{f}^\kk_{ij}(t_1)} &=
		\abs{f^{\bar{\m},\kk}_{ij}-\hat{f}^\kk_{ij}(t_1)} \nonumber \\
	&=\left(1-\frac{t_1-t^\kk_{\bar{\m}-1}}{\Delta{t}_\kk}\right)
		\abs{f^{\bar{\m},\kk}_{ij}-f^{\bar{\m}-1,\kk}_{ij}}
			\leq 2(1+\bar{\eta})(t^\kk_{\bar{\m}}-t_1).
	\label{A_1}
\end{align}

Finally, by plugging Eqs.~\eqref{A_3}--\eqref{A_1} into Eq.~\eqref{equicon} we get
$$ \abs{\hat{f}^\kk_{ij}(t_2)-\hat{f}^\kk_{ij}(t_1)}\leq 2(1+\bar{\eta})
	(t^\kk_{\bar{\m}}-t_1+t_2-t^\kk_{\bar{\n}-1}+t^\kk_{\bar{\n}-1}-t^\kk_{\bar{\m}})
		=2(1+\bar{\eta})(t_2-t_1), $$
whence
$$ \norm{\hat{\f}^\kk(t_2)-\hat{\f}^\kk(t_1)}_1\leq 2mn(1+\bar{\eta})(t_2-t_1). $$
This implies the desired equicontinuity of the family $\{\hat{\f}^\kk\}_{\kk\geq 0}$. In fact, for every $\epsilon>0$ it is sufficient to take $\abs{t_2-t_1}<\frac{\epsilon}{2mn(1+\bar{\eta})}$, independently of $\kk$, in order to have $\norm{\hat{\f}^\kk(t_2)-\hat{\f}^\kk(t_1)}_1<\epsilon$. \qedhere
\end{enumerate}
\end{proof}

\paragraph*{Step 3}
Now we are in a position to state that a mild solution to Problem~\eqref{IBVP} does indeed exist.

\begin{theorem}[Existence]
Let Assumptions~\ref{ip4b}, \ref{ip5} hold. Then $\f\in\B$ found in Lemma~\ref{lemma2} is a mild solution to Problem~\eqref{IBVP} in the sense of Definition~\ref{def:mild.solution}.
\label{theorem2}
\end{theorem}
\begin{proof}
Owing to the convergence result stated in Lemma~\ref{lemma2}, we can guess that, for a fixed $\kk$, the function $\hat{\f}^\kk$ defined componentwise by Eq.~\eqref{interpolante} is a sort of approximation of a possible mild solution to Problem~\eqref{IBVP}. Therefore, if we plug $\hat{f}^\kk_{ij}$ into Eq.~\eqref{IBVPmild} we expect a reminder $e^\kk_{ij}$ to appear:
\begin{align}
	\begin{aligned}[b]
	\hat{f}^\kk_{ij}(t)-f^0_{ij}+\int_0^t &\left \{v_j(\Phi_{i,i+1}[\hat{\f}^\kk](s)\hat{f}^\kk_{ij}(s)
		-\Phi_{i-1,i}[\hat{\f}^\kk](s)\hat{f}^\kk_{i-1,j}(s))\right. \\ 
	& \left.-G_{ij}[\hat{\f}^\kk,\,\hat{\f}^\kk](s)+\hat{f}^\kk_{ij}(s)L_{ij}[\hat{\f}^\kk](s)\right\}\,ds
		=\int_0^t e^\kk_{ij}(s)\,ds.
	\end{aligned}
	\label{int:sol}
\end{align}

At this point we would like to take the limit $\kk\to\infty$ using Lemma~\ref{lemma2}. This will tell us that the distribution function $\f$ found in the latter is indeed a mild solution to Problem~\eqref{IBVP} provided:
\begin{enumerate}
\item[(i)] the left-hand side of Eq.~\eqref{int:sol} tends to the corresponding expression evaluated for $\f$ rather than for $\hat{\f}^\kk$;
\item[(ii)] the right-hand side of Eq.~\eqref{int:sol} tends to zero.
\end{enumerate}

In order to prove (i) it is sufficient to observe that all terms appearing in the integral are bounded by an integrable constant, thus by dominated convergence it is possible to commute the limit in $\kk$ with the integral in $t$. Using Lemma~\ref{lemma2} and taking into account Eqs.~\eqref{1term}--\eqref{3term} with $\g=\hat{\f}^\kk$ yields then:
\begin{align*}
	\lim_{\kk\to\infty}\Bigl\vert\hat{f}^\kk_{ij}(t)-f^0_{ij}+\int_0^t &\left \{v_j(\Phi_{i,i+1}[\hat{\f}^\kk](s)\hat{f}^\kk_{ij}(s)
		-\Phi_{i-1,i}[\hat{\f}^\kk](s)\hat{f}^\kk_{i-1,j}(s))\right. \\ 
	& \left.-G_{ij}[\hat{\f}^\kk,\,\hat{\f}^\kk](s)+\hat{f}^\kk_{ij}(s)L_{ij}[\hat{\f}^\kk](s)\right\}\,ds\Bigr\vert \\
	= \Bigl\vert f_{ij}(t)-f^0_{ij}+\int_0^t &\left \{v_j(\Phi_{i,i+1}[\f](s)f_{ij}(s)
		-\Phi_{i-1,i}[\f](s)f_{i-1,j}(s))\right. \\
	& \left.-G_{ij}[\f,\,\f](s)+f_{ij}(s)L_{ij}[\f](s)\right\}\,ds\Bigr\vert
\end{align*}
as desired.

It remains to prove (ii), for which we need to find an expression of the term $e^\kk_{ij}$. Using Eqs.~\eqref{discrete}, \eqref{interpolante} we write:
\begin{itemize}
\item $\displaystyle{
	\begin{aligned}[t]
		\frac{d\hat{f}^\kk_{ij}}{dt}=\sum_{\n=1}^{N_\kk} &\left \{v_j\left(\Phi_{i-1,i}[\f^{\n-1,\kk}]f^{\n-1,\kk}_{i-1,j}
			-\Phi_{i,i+1}[\f^{\n-1,\kk}]f^{\n-1,\kk}_{ij}\right)\right. \\
		&\left.+G_{ij}[\f^{\n-1,\kk},\,\f^{\n-1,\kk}]-f^{\n-1,\kk}_{ij}L_{ij}[\f^{\n-1,\kk}]\right\};
	\end{aligned}
	} $
\item $\displaystyle{
	\begin{aligned}[t]
		v_j &\left(\Phi_{i,i+1}[\hat{\f}^\kk]\hat{f}^\kk_{ij}-\Phi_{i-1,i}[\hat{\f}^\kk]\hat{f}^\kk_{i-1,j}\right) \\
		&= v_j\sum_{\n=1}^{N_\kk}\left\{\Phi_{i,i+1}[\hat{\f}^\kk]f^{\n-1,\kk}_{ij}-\Phi_{i-1,i}[\f^{\n,\kk}]f^{\n-1,\kk}_{i-1,j}
			+\frac{t-t^\kk_{\n-1}}{\Delta{t}_\kk}\Phi_{i,i+1}[\hat{\f}^\kk]\left(f^{\n,\kk}_{ij}-f^{\n-1,\kk}_{ij}\right)\right. \\
		&\phantom{=} \qquad\qquad\left.+\frac{t-t^\kk_{\n-1}}{\Delta{t}_\kk}\Phi_{i-1,i}[\hat{\f}^\kk]
			\left(f^{\n-1,\kk}_{ij}-f^{\n-1,\kk}_{i-1,j}\right)\right\}\1_{[t^\kk_{\n-1},\,t^\kk_{\n}]};
	\end{aligned}
	} $
\item $\displaystyle{
	\begin{aligned}[t]
		G_{ij}[\hat{\f}^\kk,\,\hat{\f}^\kk]=\sum_{h,k=1}^{n}\sum_{\n=1}^{N_\kk}\eta[\hat{\f}^\kk]
			&A^j_{hk}[\hat{\f}^\kk]\Biggl\{f^{\n-1,\kk}_{ih}f^{\n-1,\kk}_{ik}
				+\frac{t-t^\kk_{\n-1}}{\Delta{t}_\kk}f^{\n-1,\kk}_{ih}\left(f^{\n,\kk}_{ik}-f^{\n-1,\kk}_{ik}\right) \\
		& +\frac{t-t^\kk_{\n-1}}{\Delta{t}_\kk}f^{\n-1,\kk}_{ik}\left(f^{\n,\kk}_{ih}-f^{\n-1,\kk}_{ih}\right) \\
		& +\left(\frac{t-t^\kk_{\n-1}}{\Delta{t}_\kk}\right)^2\left(f^{\n,\kk}_{ih}-f^{\n-1,\kk}_{ih}\right)
			\left(f^{\n,\kk}_{ik}-f^{\n-1,\kk}_{ik}\right)\Biggr\}\1_{[t^\kk_{\n-1},\,t^\kk_{\n}]};
	\end{aligned}
	} $
\item $\displaystyle{
	\begin{aligned}[t]
		\hat{f}^\kk_{ij}L_{ij}[\hat{\f}^\kk]=\sum_{k=1}^{n}\sum_{\n=1}^{N_\kk}\eta[\hat{\f}^\kk]
			&\Biggl\{f^{\n-1,\kk}_{ij}f^{\n-1,\kk}_{ik}+\frac{t-t^\kk_{\n-1}}{\Delta{t}_\kk}f^{\n-1,\kk}_{ij}
				\left(f^{\n,\kk}_{ik}-f^{\n-1,\kk}_{ik}\right) \\
		&+\frac{t-t^\kk_{\n-1}}{\Delta{t}_\kk}f^{\n-1,\kk}_{ik}\left(f^{\n,\kk}_{ij}-f^{\n-1,\kk}_{ij}\right) \\
		&+\left(\frac{t-t^\kk_{\n-1}}{\Delta t_\kk}\right)^2\left(f^{\n,\kk}_{ij}-f^{\n-1,\kk}_{ij}\right)
				\left(f^{\n,\kk}_{ik}-f^{\n-1,\kk}_{ik}\right)\Biggr\}\1_{[t^\kk_{\n-1},\,t^\kk_{\n}]},
	\end{aligned}
	} $
\end{itemize}
then we take the time derivative of both sides of Eq.~\eqref{int:sol} and insert these expressions at the left-hand side. As a result, we get an expression of $e^\kk_{ij}$ in which terms featuring the difference $f^{\n,\kk}_{ij}-f^{\n-1,\kk}_{ij}$ appear along with others involving differences between pairs of flux limiters, interaction rates, and tables of games evaluated at $\f^{\n-1,\kk}$ and $\hat{\f}^\kk$, respectively. Applying estimate~\eqref{A_3} to the first ones and invoking the Lipschitz continuity of the second ones we finally obtain:
$$ \abs{e^\kk_{ij}(t)}\leq\C\Delta{t}_\kk $$
for a suitable constant $\C>0$ independent of $\kk$, hence
$$ \int_0^t\abs{e^\kk_{ij}(s)}\,ds\leq\C\Delta{t}_\kk T_\M\xrightarrow{\kk\to\infty}0 $$
which completes the proof.
\end{proof}

\begin{corollary}[Improved regularity]
The mild solution to Problem~\eqref{IBVP} is of class $C^1$ in $[0,\,T_\M]$, hence it is actually a classical solution.
\end{corollary}
\begin{proof}
The continuity of $\f$ implies that the right-hand side of Eq.~\eqref{IBVP} is continuous in $t$ for each $i=1,\,\dots,\,m$ and $j=1,\,\dots,\,n$. Hence $\frac{df_{ij}}{dt}$ is continuous.
\end{proof}

\section*{Acknowledgments}
A. Tosin was partially funded by the European Commission under the 7th Framework Program (grant No.\ 257462 HYCON2 Network of Excellence) and by the Google Research Award ``Multipopulation Models for Vehicular Traffic and Pedestrians'', 2012--2013.

\bibliographystyle{plain}
\bibliography{FlTa-discrspace}
\end{document}